\documentclass[runningheads]{llncs}

\usepackage{algorithm}
\usepackage{algpseudocode}
\usepackage{xcolor}
\usepackage{amsmath}
\usepackage{amssymb}
\usepackage{mathtools}
\usepackage{amsfonts} 
\usepackage{graphicx}
\usepackage{cite}
\usepackage{url}
\usepackage{wrapfig}

\usepackage{tikz}
\usetikzlibrary{automata,positioning,shapes,fit,calc}

\newcommand{\buchi}{B\"uchi }

\newcommand{\Spot}{\textsf{Spot}}
\newcommand{\inc}[1] {\ensuremath{\mathsf{in}({#1})}}
\newcommand{\DepSynt} {\textsf{DepSynt}}

\newcommand{\DNNF}{\textsf{DNNF}}

\newif\iffullversion
\fullversiontrue

\newcommand{\fullversiontext}[2]{%
    \iffullversion
        #1%
    \else
        #2%
    \fi
}

\begin{document}

\title{On Dependent Variables in  Reactive Synthesis\thanks{This is the full version of a conference paper published in TACAS 2024.}}

\author{
S. Akshay\inst{1} \and
Eliyahu Basa\inst{2} \and
Supratik Chakraborty\inst{1} \and
Dror Fried\inst{2}
}

\authorrunning{S. Akshay, E. Basa, S. Chakraborty, D. Fried}

\institute{
Indian Institute of Technology Bombay, Mumbai, India
\and
The Open University of Israel, Israel
}

\maketitle

\begin{abstract}
Given a Linear Temporal Logic (LTL) formula over input and output variables, reactive synthesis requires us to design a deterministic Mealy machine that gives the values of outputs at every time step for every sequence of inputs, such that the LTL formula is satisfied. In this paper, we investigate the notion of dependent variables in the context of reactive synthesis. Inspired by successful pre-processing steps in Boolean functional synthesis, we define dependent variables as output variables that are uniquely assigned, given an assignment, to all other variables and the history so far. We describe an automata-based approach for finding a set of dependent variables. Using this, we show that dependent variables are surprisingly common in reactive synthesis benchmarks. Next, we develop a novel synthesis framework that exploits dependent variables to construct an overall synthesis solution. By implementing this framework using the widely used library {\Spot}, we show that reactive synthesis using dependent variables can solve some problems beyond the reach of several existing techniques. Further, among benchmarks with dependent variables, if the number of non-dependent variables is low ($\leq 3$ in our experiments), our method is able outperform all state-of-the-art tools for synthesis.


\end{abstract}

\section{Introduction}
\label{sec:intro}

Reactive synthesis concerns the design of deterministic transducers (often Mealy or Moore machines) that generate a sequence of outputs in response to a sequence of inputs such that a given temporal logic specification is satisfied.  Ever since Church introduced the problem~\cite{church1962} in 1962, there has been a rich and storied history of work in this area over the past six decades.  Recently, it was shown that a form of pre-processing, viz. decomposing a Linear Temporal Logic (LTL) specification, can lead to significant performance gains in downstream synthesis steps~\cite{FinkbeinerGP21}. The general idea of pre-processing a specification to simplify synthesis has also been used very effectively in the context of Boolean functional synthesis~\cite{bfss,cadet,akshayCAV18,EngineeringBFS}.  Motivated by the success of one such pre-processing step, viz. identification of uniquely defined outputs, in Boolean functional synthesis, we introduce the notion of dependent outputs in the context of reactive synthesis in this paper.  We develop its theory and show by means of extensive experiments that dependent outputs abound in reactive synthesis benchmarks, and can be effectively exploited to obtain synthesis techniques with orthogonal strengths vis-a-vis existing state-of-the-art techniques.  

In the context of propositional specifications, it is not uncommon for a specification to uniquely define an output variable in terms of the input variables and other output variables.  A common example of this arises when auxiliary variables, called Tseitin variables, are introduced to efficiently convert a specification not in conjunctive normal form (CNF) to one that is in CNF~\cite{tseitin1983complexity}.  Being able to identify such uniquely defined variables efficiently can be very helpful, whether it be for checking satisfiability, for model counting or synthesis.  This is because these variables do not alter the basic structure or cardinality of the solution space of a specification regardless of whether they are projected out or not.  Hence, one can often simplify the reasoning about the specification by ignoring (or projecting out) these variables. In fact, the remarkable practical success of Boolean functional synthesis tools such as Manthan~\cite{EngineeringBFS} and BFSS~\cite{bfss,akshayCAV18} can be partly attributed to efficient techniques for identifying a large number of uniquely defined variables.  We draw inspiration from these works and embark on an investigation into the role of uniquely defined variables, or \emph{dependent variables}, in the context of reactive synthesis.  To the best of our knowledge, this is the first attempt at directly using dependent variables for reactive synthesis.

We start by first defining the notion of dependent variables in  LTL specifications for reactive synthesis.
Specifically, given an LTL formula $\varphi$ over a set of input variables $I$ and output variables $O$, a set of variables $X\subseteq O$ is said to be \emph{dependent} on a set of variables $Y\subseteq I\cup (O\backslash X)$ in $\varphi$, if at every step of every infinite sequence of inputs and outputs satisfying $\varphi$, the finite history of the sequence together with the current assignment for $Y$ uniquely defines the current assignment for $X$. The above notion of dependency generalizes the notion of uniquely defined variables in Boolean functional synthesis, where the value of a uniquely defined output at any time is completely determined by the values of inputs and (possibly other) outputs at that time.  We show that our generalization of dependency in the context of reactive synthesis is useful enough to yield a synthesis procedure with improved performance vis-a-vis competition winning tools, for a non-trivial number of reactive synthesis benchmarks.

We present a novel automata-based technique for identifying a subset-maximal set of dependent variables in an LTL specification $\varphi$. Specifically, we convert $\varphi$ to a language-equivalent non-deterministic \buchi automaton (NBA) $A_\varphi$, and then deploy efficient path-finding techniques to identify a subset-maximal set of outputs $X$ that are dependent on $Y = I \cup (O \setminus X)$. We implemented our method to determine the prevalence of dependent variables in existing reactive synthesis benchmarks. Our finding shows that out of 1141 benchmarks taken from the SYNTCOMP~\cite{syntComp} competition, 300 had at least one dependent output variable and 26 had all output variables dependent.

Once a subset-maximal set, say $X$, of dependent variables is identified, we proceed with the synthesis process as follows. Referring to the NBA $A_\varphi$ alluded to above, we first transform it to an NBA $A_\varphi'$ that accepts the language $L'$ obtained from $L(\varphi)$ after removing (or projecting out) the $X$ variables.  Our experiments show that $A_\varphi'$ is more compactly representable compared to $A_\varphi$, when using BDD-based representations of transitions (as is done in state-of-the-art tools like \Spot~\cite{spot}). Viewing $A_\varphi'$ as a new (automata-based) specification with output variables $O \setminus X$, we now synthesize a transducer $T_Y$ from $A'$ using standard reactive synthesis techniques.  This gives us a strategy $f^Y:\Sigma_I^*\rightarrow \Sigma_{O\backslash X}$ for each non-dependent variable in $O\setminus X$.  Next, we use a novel technique based on Boolean functional synthesis to directly construct a circuit that implements a transducer $T_X$ that gives a strategy $f_X:\Sigma_Y^*\rightarrow \Sigma_X$ for the dependent variables.  Significantly, this circuit can be constructed in time polynomial in the size of the (BDD-based) representation of $A_\varphi$. The transducers $T_Y$ and $T_X$ are finally merged to yield an overall transducer $T$ that describes a strategy $f:\Sigma_I^*\rightarrow \Sigma_O$ solving the synthesis problem for $\varphi$. 

We implemented our approach in a tool called {\DepSynt}.  Our tool was developed in C++ using APIs from the widely used library {\Spot} for representing and manipulating non-deterministic B\"{u}chi automata. We performed a comparative analysis of our tool with winning entries of the SYNTCOMP~\cite{syntComp} competition to evaluate how knowledge of dependent variables helps reactive synthesis. Our experimental results show that identifying and utilizing dependent  variables results in improved synthesis performance when the count of non-dependent variables is low.  Specifically, our tool outperforms state-of-the-art and highly optimized synthesis tools on benchmarks that have at least one dependent variable and at most 3 non-dependent variables. This leads us to hypothesize that exploiting dependent variables benefits synthesis when the count of non-dependent variables is below a threshold.  Given the preliminary and un-optimized nature of our implementation, we believe there is significant scope for improvement of our results. 

\paragraph{Related work}
\label{sec:related}

Reactive Synthesis has been an extremely active research area for the last several decades (see e.g. \cite{church1962,pnueli89,BloemJPPS12,FinkbeinerS13,FinkbeinerGP21}).  Not only is the theoretical investigation of the problem rich, there are also several tools that are available to solve synthesis problems in practice. These include solutions like {\tt ltlsynt}~\cite{michaud2018reactive} based on {\Spot}~\cite{spot}, Strix~\cite{meyer2018strix} and BoSY~\cite{faymonville2017bosy}. Our tool relies heavily on  {\Spot} and its APIs, which we use liberally to manipulate non-deterministic \buchi automata. Our synthesis approach is based on the standard conversion of the converted NBA to a deterministic parity automata (DPA) (see~\cite{BloemCJ18} for an overview of the challenges of reactive synthesis).

Our work may be viewed as lifting the idea of uniquely defined variables in Boolean functional synthesis to the context of reactive synthesis.  Viewed from this perspective, our work is not the first to lift ideas from Boolean functional Synthesis. Following a framework called "back-and-forth" for Boolean functional synthesis that uses a decomposition of a specification into separate formulas on input variables and on output variables~\cite{ChakrabortyFTV22}, the work in~\cite{AmramBFTVW21} constructed a reactive synthesis tool for specific benchmarks that admit a separation of the specification into formulas for only the environment variables and formulas for only the system variables. In this sense, this current work of ours is another step in bridging the gap between the  Boolean functional synthesis and the reactive synthesis communities.

The remainder of the paper is structured as follows. We  introduce definitions and notations in Section~\ref{sec:prelims}. In Section~\ref{sec:dependencies} we define dependent variables for LTL formulas, and describe our framework to find it. In Section~\ref{sec:synt} we describe our automata-based synthesis framework and discuss some of its implementation details in Section~\ref{sec:ImpNew}. We describe our evaluation in Section~\ref{sec:eval} and conclude in Section~\ref{sec:conc}.

\section{Preliminaries}
\label{sec:prelims}

Given a finite alphabet $\Sigma$, an infinite \emph{word} $w$ is a sequence $w_0w_1w_2\cdots$ where for every $i$, the $i^{th}$ letter of $w$, $w_i\in\Sigma$.
The \emph{prefix}  $w_0\cdots w_i$ of size $i+1$ of $w$ is denoted by $w[0,i]$. Then $w[0,0]=w_0$. We denote $w[0,-1]$ to be the empty word. The set of all infinite words over $\Sigma$ is denoted by $\Sigma^\omega$. We call $L\subseteq \Sigma^\omega$ a \emph{language} over infinite words in $\omega$.
When the alphabet $\Sigma$ is combined from two distinct alphabets, $\Sigma=\Sigma_X\times\Sigma_Y$ for some sets of variables $X,Y$, we abuse notation and for a letter $a=(a_1,a_2)\in\Sigma$, denote by $a.X$ the projection of $a$ on $\Sigma_X$, that is, the letter $a_1\in\Sigma_X$. Similarly, $a.Y$ denotes the projection of $a$ on $\Sigma_Y$, that is the letter $a_2\in\Sigma_Y$. We define by $w.X$ the word obtained from a word $w\in\Sigma$ on $\Sigma_X$ i.e. $w.X= w_0.X\cdots$. 

\paragraph{Linear Temporal Logic.}

A Linear Temporal Logic (LTL) formula is constructed with a finite set of propositional variables $V$, using Boolean operators such as $\vee, \wedge,$ and $\neg,$ and temporal operators such as next $X$, until $U$, etc. 
 The set $V$ induces an alphabet $\Sigma_V=2^V$ of all possible assignments ($true/false$) to the variables of $V$.
The language of an LTL formula $\varphi$, denoted $L(\varphi)$ is the set of all words in $\Sigma_V^\omega$. The semantics of the operators and satisfiability relation is defined as usual~\cite{huth-ryan}. We denote the number of variables $V$ by $|V|$, and size of formula $\varphi$, i.e., number of its subformulas by $|\varphi|$. We sometimes abuse notation and identify the singleton set of a variable $\{z\}$ with the variable $z$, and denote $E_V$ by $E$ when $V$ is clear from the context.

\paragraph{Nondeterministic \buchi Automata.}
A Nondeterministic \buchi Automaton (NBA) is a tuple $A=(\Sigma, Q, \delta, q_0, F)$ where $\Sigma$ is the alphabet, $Q$ is a finite set of states, $\delta: Q \times \Sigma \rightarrow 2^Q$ is a non-deterministic transition function, $q_0$ is the initial state and $F\subseteq Q$ is a set of accepting states. $A$ can be seen as a  directed labeled graph with vertices $Q$ and an edge $(q,q')$ exists with a label $a$ if $q'\in \delta(q,a)$. We denote the set of incoming edges to $q$ by $in(q)$ and the set of outgoing edges from $q$ by $out(q)$.
A \emph{path} in $A$ is then a (possibly infinite) sequence of states $\rho=(q_{i_0},q_{i_1},\cdots)$ in which for every $j>0$, $(q_{i_j},q_{i_{j+1}})$ is an edge in $A$. A \emph{run} is a path that starts in $q_0$, and is \emph{accepting} if it visits a state in $F$ infinitely often. A \emph{word} of the run $\rho$ is the sequence of labels seen along $\rho$, i.e., $w=\sigma_{i_0}\sigma_{i_1}\cdots$ where for every $j\geq 0$, $q_{i_{j+1}}\in \delta(q_{i_{j}}, \sigma_{i_j})$. As $A$ is nondeterministic a word can have many runs, although every run has a single word. A word is \emph{accepting} if it has an accepting run in $A$. The language $L(A)$ is the set of all accepting words in $A$. Wlog, we assume that all states and edges that are not a part of an accepting run are removed. 
Finally, every LTL formula $\varphi$ can be transformed in exponential time in the length of $\varphi$, to an NBA $A_\varphi$ for which $L(\varphi)=L(A_\varphi)$~\cite{VW94,huth-ryan}. When $\varphi$ is clear from  the context we omit the subscript and refer to $A_\varphi$ as $A$. We denote by $|A|$ the size of an automaton, i.e., its number of states and transitions.

\paragraph{Reactive Synthesis.} A \emph{reactive LTL formula} is an LTL formula $\varphi$ over a set of input variables $I$ and output variables $O$, with $I\cap O=\emptyset$.
In \emph{reactive synthesis} we are given a reactive LTL formula $\varphi$, and the challenge is to synthesize a function, called \emph{strategy}, $f:\Sigma_I^*\rightarrow \Sigma_O$ such that every word $w\in(\Sigma_I\times\Sigma_O)^\omega$ obtained by using this strategy is in $L(\varphi)$. If such a strategy exists we say that $\varphi$ is \emph{realizable}. Otherwise, we say that $\varphi$ is \emph{unrealizable}.  In what follows, we always consider only reactive LTL formulas and hence often omit the prefix reactive while referring to them. The synthesized strategy $f:\Sigma_I^*\rightarrow \Sigma_O$ that is given as output is typically described (explicitly or symbolically) in a form of a \emph{transducer} $T=(\Sigma_I,\Sigma_O,S,s_0,\delta,\lambda)$ in which $\Sigma_I$ and $\Sigma_O$ are input and output alphabet respectively, $S$ is a set of states with an initial state $s_0$, $\delta:S\times\Sigma_I\rightarrow S$ is a deterministic transition function, and $\lambda:S\times\Sigma_I\rightarrow \Sigma_O$ is the output function. A standard procedure in solving reactive synthesis is to transform the given LTL formula $\varphi$ to an NBA $A_\varphi$ for which $L(A_\varphi)=L(\varphi)$.
Then transform $A_\varphi$ to a Deterministic Parity Automata (DPA) that turns to a parity game, which solution is described as a transducer $T_{A_\varphi}$. 
It is known that this approach cannot escape exponential blowups. More precisely,
\begin{theorem}
\begin{enumerate}
    \item Reactive synthesis can be solved in $O(2^{n\cdot 2^{n}})$, where $n$ is the size of the LTL formula.
    \item Given an NBA $A$ with $n$ states, computing transducer $T_A$ takes $\Omega (2^{n\log n})$.
    \end{enumerate} 
\end{theorem}

\section{Dependent variables in reactive LTL}
\label{sec:dependencies}
In this section, we define dependent variables for (reactive) LTL formulas and suggest a framework for finding a maximal set of dependent variables. Our notion of dependent variables for LTL formulas, specifically suited to reactive synthesis, asks that the dependency be maintained at every step of the word satisfying the formula. While there are several notions of dependency that can be considered, we describe the one that we use throughout the paper. As mentioned earlier, our definition of dependency is restricted to output variables, since having dependent input variables would imply that not all input values are possible for these variables, which makes the formula unrealizable.

\begin{definition}[Variable Dependency in LTL]\label{def:LTLDep}
Let $\varphi$ be an LTL formula over $V$ with input variables $I\subseteq V$ and output variables $O=V\backslash I$. Let $X,Y$ be sets of variables where $X\subseteq O$. We say that $X$ is dependent on $Y$ in $\varphi$ if for every pair of words $w,w'\in L(\varphi)$ and $i\geq 0$ if $w[0,i-1]=w'[0,i-1]$ and $w_i.Y=w'_i.Y$, then we have  $w_i.X = w'_i.X$. Further, we say that $X$ is
dependent in $\varphi$ if $X$ is dependent on $V\setminus{X}$ in $\varphi$, i.e., it is dependent on all the remaining variables.
\end{definition}
Note that two words in $L(\varphi)$ with different prefixes can still have different values for $X$ for the same values for $Y$, and $X$ can still be defined dependent on $Y$. 

As an example, consider an LTL formula $\varphi$ with an input variable $y$, an output variable $x$ and a language $L= \{w^1,w^2,w^3\}$ where $w^1= (y,x)^\omega$, $w^2= (\neg y,x)^\omega$ and $w^3=(y,x)(\neg y,x)(y,\neg x)^\omega$. Then $x$ is dependent on $y$ in $\varphi$. Specifically note that $w^1[0,1]\neq w^3[0,1]$ and thus the dependency of $x$ is not violated, although $w^1_2.y=w^3_2.y$ and $w^1_2.x\not =w^3_2.x$.
Observe that, if $X$ is dependent on $Y$ in $\varphi$ for some $Y$, then it is also dependent in $\varphi$. We next show how to find a maximal set of dependent variables.

\subsection{Maximally dependent sets of variables}
\label{sec:findDep}
\begin{definition} 
Given an LTL formula $\varphi(I,O)$, we say that a set $X\subseteq O$ is a \emph{maximal dependent set} in $\varphi$ if $X$ is dependent in $\varphi$ and every set that strictly contains $X$ is not dependent in $\varphi$. 
\end{definition}
As in the propositional case~\cite{SM22}, finding maximum or minimum dependent sets is intractable, hence we focus on subset-maximality. Given a variable $z$ and set $Y$, checking whether $z$ is dependent on $Y$, can easily be used to finding a maximal dependent set. Indeed, we would just need to start from the empty set and iterate over output variables, checking for each if it is dependent on the remaining variables. For completeness, we give the pseudocode for this in \fullversiontext{Appendix~\ref {sec:appdep}}{full-version paper \cite{onDepVarFullVersion}}.

Note that when all output variables are not dependent, the order in which output variables are chosen may play a significant role in the size of the maximal set obtained. We currently use a naive ordering (first appearance), and leave the problem of better heuristics for getting larger maximal independent sets to future work.

\subsection{Finding dependent variables via  automata}\label{sec:autdep}
 As explained above, the heart of the dependency check is to check if a given output variable is dependent on a set of variables. We now develop an approach for doing so based on the nondeterministic \buchi automaton $A_\varphi$ of the original LTL formula $\varphi$. Our framework uses the notion of compatible pairs of states of the automaton  defined as follows.

\begin{definition}\label{def:comptStates}
Let $A=(\Sigma, Q, \delta, q_0, F)$ be an NBA with states $s,s'$ in $Q$. Then the pair $(s,s')$ is \emph{compatible} in $A$ if there are runs from $q_0$ to $s$ and from $q_0$ to $s'$ with the same word $w\in\Sigma^*$.  
\end{definition}
Recall that in our definition, only states and edges that are part of an accepting run exist in $A$. Then we have the following definition.

\begin{definition}\label{def:autdep}
Let $\varphi$ be an LTL formula over $V$ with input variables $I\subseteq V$ and output variables $O=V\backslash I$. Let $X,Y$ be sets of variables where $X\subseteq O$.
Let $A_\varphi$ be an NBA that describes $\varphi$. 
We say that $X$ is \emph{automata dependent} on $Y$ in $A_\varphi$, if for every pair of compatible states $s,s'$ and assignments $\sigma$, $\sigma'$ for $V$, where $\sigma.Y=\sigma'.Y$ and $\sigma.X\neq \sigma'.X$, $\delta(s,\sigma)$ and $\delta(s,\sigma')$ cannot both exist in $A_\varphi$. We say that $X$ is automata dependent in $A_\varphi$ if $X$ is automata dependent on $Y$ in $A_\varphi$ and $Y=V\backslash X$.
\end{definition}
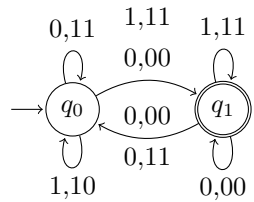
\begin{wrapfigure}[10]{r}{0.35\textwidth}
\begin{tikzpicture}[shorten >=1pt,node distance=2cm,on grid,auto]
  \tikzset{every state/.style={minimum size=0pt}}

  \node[state,initial,initial text={}]  (q_0)   {$q_0$};
  \node[state,accepting]          (q_1) [right=of q_0] {$q_1$};

  \path[->] (q_0) edge [loop above]  node        {0,11} (q_0)
                  edge [bend left]   node [above,swap,label={1,11}] {0,00} (q_1)
                  edge [loop below]  node        {1,10} (q_0)
            (q_1) edge [bend left]  node [label={0,00}] {0,11} (q_0)
                  edge [loop above] node        {1,11} (q_1)
                  edge [loop below]  node       {0,00} (q_1);
\end{tikzpicture}
\caption{An Example NBA $A_1$}
\label{fig:example}
\end{wrapfigure}
As an example, consider NBA $A_1$ from adjoining Figure~\ref{fig:example}, constructed from some LTL formula with input $I=\{i\}$ and outputs $O=\{o_1,o_2\}$. Here $\Sigma_I=\{0,1\}, \Sigma_O=\{0,1\}^2$ and edges are labeled by values of $(i, o_1o_2)$. It is immediate that, $(q_0,q_0), (q_1,q_1)$ are compatible pairs  but so are $(q_0,q_1), (q_1,q_0)$ since they can both be reached from the initial state on reading the length $2$ word $(0,00)(0,00)$. Now consider output $o_1$. It is not dependent on $\{i\}$, i.e., only the input, since from $q_0$ with $i=0$, we can go to different states with different values of $o_1$. But $o_1$ is indeed dependent on $\{i, o_2\}$. To see this consider every pair of compatible states -- in this case all pairs. Then we can see that if we fix the values of $i$ and $o_2$, there is a unique value of $o_1$ that permits state transitions to happen from the compatible pair.  For example, regardless of which state we are in, if $i=0, o_2=0$, $o_1$ must be $0$ for a state transition to happen. On the other hand, $o_2$ is not dependent on either $\{i\}$ or $\{i, o_1\}$ (as can be seen from $(q_0,q_1)$ with $i=1,o_1=1$).

The following shows us the relation between automata dependency and dependency in LTL as defined earlier.

\begin{theorem}\label{thm:autdep}
Let $\varphi$ be an LTL formula with $V=I\cup O$ variables, where $X\subseteq O$ and  $Y\subseteq I\cup O$. Let $A_\varphi$ be an NBA that describes $\varphi$. Then $X$ is dependent on $Y$ in $\varphi$ if and only if $X$ is automata dependent on $Y$ in $A_\varphi$.
\end{theorem}

\begin{proof}

 Assume that $X$ is dependent in $\varphi$ and let $s,s’$ be two compatible states in $A$ with a joint history $w^*$. Assume for contradiction that there is an assignment $y$ for $Y$, for which there are distinct assignments $x,x’$ for $X$ for which both $\delta(s,xy)$ exists and $\delta(s’,x’y)$ exist. Then this means that $\delta(s,xy)$ and $\delta(s’,x’y)$ both continue to some accepting runs $r,r$ respectively in $A$. Let $w$ be the label of the $r$ and $w’$ the label of $r’$. Then $w,w’$ are in $L(\varphi)$ with a prefix $w^*$, but $x=w_i.x\neg=w’_i.x=x’$ contradicting $X$ being relaxed dependent in $\varphi$.

 Next, assume that $X$ is automata dependent on $Y$ in $A_\varphi$. Let $w,w'$ be two words in $L(\varphi)$, and thus in $L(A_\varphi)$. Assume that there is $i\geq 0$ for which $u=w[0,i]=w'[0,i]$. Let $r, r'$ be accepting runs of $w,w'$ respectively and let $s_1,s_2$ be the states respectively in the run $r,r'$ that has the same history $w[0,i-1]$. Note that this makes $s_1,s_2$ compatible. Next assume $e_1=(s_1,s_1')$ with a label $\sigma_1$ is on $r$ , and $e_2=(s_2,s_2')$  with a label $\sigma_2$ is on $r'$. Further, assume $\sigma_1.y=\sigma_2.y$. Note that it means that $w_{i+1}.Y=w'_{i+1}.Y$. Then since $X$ is automata dependent in $Y$, we have that $\sigma_1.X=\sigma_2.X$ as well, hence $w_{i+1}.X=w'_{i+1}.X$.\qed
\end{proof}

Thus, we can focus on algorithms for finding automata dependency. We start by describing how to find compatible states.

\paragraph{Finding Compatible States} We find all compatible states in an automaton in Algorithm \ref{alg:all-compatible-states} as follows. We maintain a list of in-process compatible pairs $C$ to which we start by adding the initial pair $(q_0,q_0)$, which is of course compatible. At each step, until $C$ becomes empty, for each pair $(s_i,q_j)\in C$, we add it to the compatible pair set $P$, and remove it from $C$ (in line 4). Then (in lines 5-8), we check (in line 6) if outgoing transitions from $(s_i,s_j)$ lead to a new pair $(s'_i,s'_j)$ not already in $P$ or $C$, which can be reached on reading the same letter $\sigma$ and if so, we add this pair to the in-process set $C$. All pairs that we put in $P,C$ are indeed compatible, nothing is removed from $P$. And when the algorithm terminates, $C$ is empty, which means all possible ways (from initial state pair) to reach a possible compatible pair have been explored, thus showing correctness. 

\begin{algorithm}[htb!]
\caption{Find All Compatible States in NBA}\label{alg:all-compatible-states}
\hspace*{\algorithmicindent} \textbf{Input}
NBA $A_\varphi = (\Sigma, Q, \delta, q_0, F)$ of $\varphi$.\\
\hspace*{\algorithmicindent} \textbf{Output}
Set $P\subseteq Q \times Q$ of all compatible state pairs in $A_\varphi$
\begin{algorithmic}[1]
\State $P \gets \emptyset$; ~$C \gets  \{(q_0,q_0)\}$
\While{$C \neq \emptyset$}
    \State Let $(s_i, s_j) \in C$
    \State $P \gets P \cup \{(s_i,s_j)\}$; ~$C \gets C \setminus \{(s_i, s_j)\}$
    \For{$(s_i', s_j') \in out(s_1) \times out(s_2)$}
        \If{$(s_i',s_j') \notin P \cup C$ and $\exists \sigma \in 2^\Sigma$ s.t. $s_i' \in \delta(s_i, \sigma)  \wedge s_j' \in \delta(s_j, \sigma)$} \label{line:alg-compat-states-check}
            \State $C \gets C \cup \{(s_i',s_j')\}$
        \EndIf
    \EndFor
\EndWhile
\State \textbf{return} $P$    
\end{algorithmic}
\end{algorithm}

Finally, we show how to check dependency using automata, by implementing the following procedure \textbf{isAutomataDependent}, described in Algorithm~\ref{alg:find-deps-by-automaton}. \textbf{isAutomataDependent} works by trying to find a witness to  $\{z\}$ being \emph{not} dependent on $Y$. If no such witness exists then it means that $\{z\}$ is dependent on $Y$. Given a variable $x$ and a set $Y=V\backslash \{z\}$,\textbf{isAutomataDependent} first uses Algorithm~\ref{alg:all-compatible-states} to construct a list $P$ of all compatible pairs in $A$ (line 4). Then for every pair $(s,s')\in P$ the algorithm checks through the procedure $\textsf{AreStatesColliding}$ (lines 1-2) whether there exists an assignment $\sigma,\sigma'$ for which both $\delta(s,\sigma)$ and $\delta(s',\sigma')$ exist, $\sigma. Y=\sigma'. Y$ and $\sigma. \{z\}\neq \sigma'.\{z\}$. If so, then the algorithm returns \emph{false}: $\{z\}$ is not dependent on $Y$ (line 7). At the end of checking all the pairs, the algorithm returns \emph{true}.

\begin{algorithm}[htb!]
\caption{Check Dependency Based Automaton}\label{alg:find-deps-by-automaton}
\hspace*{\algorithmicindent} \textbf{Input} 
\hspace*{\algorithmicindent} NBA $A_\varphi = (\Sigma, Q, \delta, q_0, F)$ from $\varphi$, 
Candidate dependent variable $\{z\}$, 
Candidate dependency set $Y$.\\
\hspace*{\algorithmicindent} \textbf{Output}
Is $\{ z \}$ dependent on $Y$ by Definition \ref{def:autdep}
\begin{algorithmic}[1]
\Procedure{AreStateColliding}{$p, q$}
    \State \textbf{return} $\exists \sigma_p, \sigma_q \in 2^\Sigma$ s.t.
        $ \delta(p, \sigma_p) \neq \emptyset \wedge \delta(q, \sigma_q) \neq \emptyset \wedge \sigma_p.Y = \sigma_q.Y \wedge \sigma_p.\{ z \} \neq \sigma_q.\{ z \}$
\EndProcedure

\State $P \gets FindAllCompatibleStates(A_\varphi)$ 
\For{$(s_1, s_2) \in P$}
    \If{$AreStateColliding(s_1, s_2)$}
        \State \textbf{return} False
    \EndIf
\EndFor

\State \textbf{return} True
\end{algorithmic}
\end{algorithm}

\begin{lemma}
    \label{clm:depAutCorrect}
    Algorithm~\ref{alg:find-deps-by-automaton} returns \emph{True} if and only if $\{z\}$ is automata-dependent on $V\backslash \{z\}$ in $A_\varphi$.
\end{lemma}
Thus using the above algorithm to perform dependency check we can compute maximal sets of dependent variables (as explained earlier), which we will use next to improve synthesis. Note that all the above algorithms run in time polynomial (in fact, quadratic) in size of the NBA. 
\begin{corollary} Given NBA $A_\varphi$, computing compatible pairs, checking dependency, building maximal dependent sets can be done in time polynomial in the size of $A_\varphi$.
\end{corollary}

Note that if all output variables are dependent, then regardless of the order in which the outputs are considered, for every finite history of inputs, there is a unique value for each output that can cause the specification to be satisfied. Therefore, there is a unique winning strategy for the specification, assuming it is realizable.

\section{Exploiting Dependency in Reactive Synthesis}\label{sec:synt}

\begin{figure}[t]
\centering
\scalebox{0.84}{
\begin{tikzpicture}
\node[] (a) at (-3.7,0) {$\varphi$};
\node[draw, thick, fill=gray!40!white, minimum height=1cm] (ltl) at (-2,0) {1. LTL to NBA};
\node[] (b) at (-0.1,0) {$A_\varphi$};
\node[draw, thick, minimum height=1cm] (dep) at (1.9,0) {2. Identify Dep};
\node[draw, thick, minimum height=1cm] (proj) at (5.5,1) {3. $Proj_{dep}$};
\node[draw, thick, fill=gray!40!white,minimum height=1cm] (ty) at (8.7,1) {4. Syn-Nondep $T_Y$};
\node[draw, thick, minimum height=1cm] (tx) at (5,-1) {5. Syn-Dep $T_X$};
\node[draw, thick, minimum height=1cm,fill=gray!40!white] (merge) at (8.5,-1) {6. Syn-Comb $T$};
\node[] (f) at (10.5,-1) {$f$};

\path[draw,->,thick] (a) -- (ltl);
\path[draw,->,thick] (ltl) -- (b);
\path[draw,->,thick] (b) -- (dep);
\path[draw,->,thick] (dep) --node[above]{$A_\varphi,X,Y$} (proj);
\path[draw,->,thick] (proj) --node[above]{$A'_\varphi,Y$} (ty);
\path[draw,->,thick] (dep) --node[above]{$A_\varphi,X$} (tx);
\path[draw,->,thick] (ty) --node[left]{$f_Y$} (merge);
\path[draw,->,thick] (tx) --node[above]{$f_X$} (merge);
\path[draw,->,thick] (merge) -- (f);


 \end{tikzpicture}
 }
 \label{fig:pipeline}
 \caption{Synthesis using dependencies}
 \vspace*{-3mm}
 \end{figure}
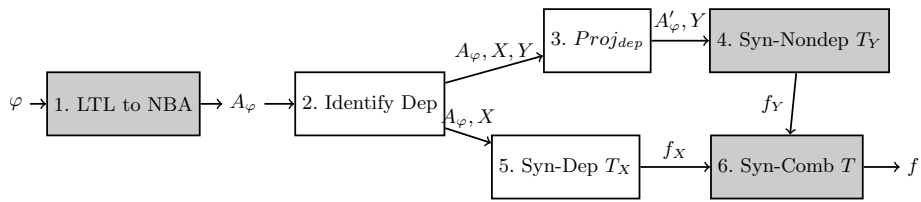
 In this section, we explain how dependencies can be beneficially exploited in a reactive synthesis pipeline.
Our approach can be described at a high level as shown in Figure~\ref{fig:pipeline}.  This flow-chart has the following 6 steps:
\begin{enumerate}
    \item Given an LTL formula $\varphi$ over a set of variables $V$ with input variables $I\subseteq V$ and output variables $O=V\backslash I$, we first construct a language-equivalent NBA $A_\varphi=(\Sigma_I\cup\Sigma_O, S,s_0, \delta, F)$ by standard means, e.g~\cite{VW94}. 
    \item Then, as described in Section \ref{sec:dependencies}, we find in $A_\varphi$ a maximal set of output variables $X$ that are dependent in $\varphi$. For
    notational convenience, in the remainder of the discussion, we use $Y$ for $I\cup (O\backslash X)$ and $\Sigma_Y$ for $\Sigma_I \times \Sigma_{O\setminus X}$.
\item Next, we construct an NBA $A'_\varphi$ from $A_\varphi$ by projecting out (or eliminating) all $X$ variables from labels of transitions. Thus, $A_\varphi'$ has the same sets of states and transitions as $A_\varphi$. We simply remove valuations of variables in $X$ from the label of every state transition in $A_\varphi$ to obtain $A_\varphi'$.  Note that after this step,  $L(A_\varphi') = \{w \mid \exists u \in L(A_\varphi) \mbox{ s.t. } w =  u.Y\} \subseteq \Sigma_Y^\omega$.
\item Treating $A'_\varphi$ as a (automata-based) specification with inputs $I$ and outputs $O\setminus X$, we next use existing reactive synthesis techniques (e.g.,~\cite{BloemCJ18}) to obtain a transducer $T_Y$ that describes a strategy $f_Y:\Sigma_I^*\rightarrow \Sigma_{O\backslash X}$ for $L(A'_\varphi)$.
\item We also construct a transducer $T_X$ that describes a function $f_X:(\Sigma_Y^*\rightarrow \Sigma_X)$ with the following property: for every word $w'\in L(A'_\varphi)$ there exists a unique word $w\in L(\varphi)$ such that $w.Y=w'$ and for all $i$, $w_i.X=f_X(w'[0,i])$.
\item Finally, we compose $T_X$ and $T_Y$ to construct a transducer $T$ that defines the final strategy $f:\Sigma_I^*\rightarrow \Sigma_O$. Recall that transducer $T_Y$ has $I$ as inputs and $O\setminus X$ as outputs, while transducer $T_X$ has $I$ and $O\setminus X$ as inputs and $X$ as outputs.  Composing $T_X$ and $T_Y$ is done by simply connecting the outputs $O\setminus X$ of $T_Y$ to the corresponding inputs of $T_X$.
\end{enumerate}

In the above synthesis flow, we use standard techniques from the literature for Steps 1 and 4, as explained above.  Hence we do not dwell on these steps in detail. Step 2 was already described in detail in Section~\ref{sec:dependencies}. Step 3 is easy when we have an explicit representation of the automata.  As we will discuss in the next section, it has interesting consequences when we use symbolic representations of automata. Step 6, as explained above, is also easy to implement.  Hence, in the remainder of this section, we focus on Step 5, which is also a key contribution of this paper. In the next section, we will discuss how steps 2, 3 and 5 are implemented using symbolic representations (viz. ROBDDs).

\subsubsection{Constructing transducer $T_X$}

Let $A = (\Sigma_I\times \Sigma_O, Q, \delta, q_0, F)$ be the NBA $A_\varphi$ obtained in step 1 of the pipeline shown above. Since each letter in $\Sigma_O$ can be thought of as a pair $(\sigma, \sigma')$, where $\sigma \in \Sigma_{O\setminus X}$ and $\sigma' \in
\Sigma_{X}$, the transition function $\delta$ can be viewed as a map from $Q \times (\Sigma_I \times \Sigma_{O\setminus X} \times \Sigma_X)$ to $2^Q$. The transducer $T_X$ we wish to construct is a deterministic Mealy machine described by the $6$-tuple $(\Sigma_Y, \Sigma_X \cup \{\bot\}, Q^X, q^X_0, \delta^{X}, \lambda^{X})$, where $\Sigma_Y = \Sigma_I \times \Sigma_{(O\setminus X)}$ is the input alphabet, $\Sigma_X$ is the output alphabet with $\bot\not\in \Sigma_X$ being a special symbol that is output when no symbol of $\Sigma_X$ suffices, $Q^X=2^Q$, that is the powerset of $Q$ is the set of states of $T_X$, $q_0^X= \{q_0\}$ is the initial state, $\delta^{X}: Q^X \times
\Sigma_I \times \Sigma_{(O \setminus X)} \rightarrow Q^X$ is the state transition function, and $\lambda^{X}: Q^X \times \Sigma_I \times
\Sigma_{(O \setminus X)} \rightarrow \Sigma_X$ is the output function.
The state transition function $\delta^{X}$ is defined by the Rabin-Scott subset construction applied to the automaton $A_\varphi$~\cite{HopcroftUllman79}. Formally, for every $U \subseteq Q$, $\sigma_I \in \Sigma_I$ and $\sigma \in \Sigma_{(O \setminus X)}$, we define $\delta^{X}\big(U, (\sigma_I, \sigma)\big) = \{q' \mid q' \in Q,\, \exists q \in U$ and $\exists
\sigma' \in \Sigma_X$ s.t. $q' \in \delta\big(q, (\sigma_I, \sigma, \sigma')\big)\}$.  
Before defining the output function $\lambda^{X}$, we state an important property of $T^X$ that follows from the definition of $\delta^X$ above.
\begin{lemma}\label{lem:trans-states-compat}
If $X$ is automata dependent in $A_\varphi$, then every state $U$ reachable from $q^X_0$ in $T_X$ satisfies the property: $\forall q, q' \in U$, $(q,q')$ is compatible in $A_\varphi$.  
\end{lemma}

\begin{proof}
We prove by induction on the number of steps, say $k$, needed to reach a state $U$ from $q^X_0$ in $T_X$. The base case follows from the fact that $q^X_0$ is a singleton set, and every singleton set trivially satisfies the desired property.  Therefore the claim holds for $k=0$.

Suppose the claim holds for all states $U$ reachable from $q^X_0$ in $k$ or fewer steps, for $k \ge 0$.  Hence, $(q, q')$ is compatible in $A_\varphi$ for all $q, q' \in U$.  We wish to prove that if $U' = \delta^X\big(U, (\sigma_I, \sigma)\big)$ for any $(\sigma_I, \sigma) \in \Sigma_I \times \Sigma_{O\setminus X}$, then $(s, s')$ is also compatible in $A_\varphi$ for every $s, s' \in U'$. 

Since $X$ is automata dependent in $A_\varphi$, it follows from 
Definition~\ref{def:autdep} that for every $(q, q') \in U \times U$, if $\delta\big(q, (\sigma_I, \sigma, \sigma_X)\big) \neq \emptyset$ and $\delta\big(q', (\sigma_I, \sigma, \sigma_X')\big) \neq \emptyset$ for some $\sigma_X, \sigma_X' \in \Sigma_X$, then $\sigma_X = \sigma_X'$.  From the definition of $\delta^X$, it now follows that if $U' = \delta^X\big(U, (\sigma_I, \sigma)\big)$ and if $U' \neq \emptyset$, then there exists a unique $\sigma_X \in \Sigma_X$ s.t. $U' = \{q' \mid \exists q \in U \mbox{ s.t. } q' \in \delta\big(q, (\sigma_I, \sigma,\sigma_X)\big) \}$.
This shows that for every $s, s' \in U'$, the pair $(s, s')$ is compatible in $A_\varphi$.\qed
\end{proof}

We are now ready to define the output function $\lambda^X$ of $T_X$.  Let $U$ be a state reachable from $q^X_0$ in $T_X$ and let $U' = \delta^{X}\big(U, (\sigma_I, \sigma)\big)$, where $(\sigma_I, \sigma) \in \Sigma_Y$.  If $U' \neq \emptyset$, we know from the proof of Lemma~\ref{lem:trans-states-compat}  that there is a unique $\sigma_X \in \Sigma_X$ s.t. $U' = \{q' \mid \exists q \in U \mbox{ s.t. } q' \in \delta\big(q, (\sigma_I, \sigma,\sigma_X)\big) \}$.  We define
$\lambda^{X}\big(U, (\sigma_I, \sigma)\big) = \sigma_X$ in this case.  If, on the other hand, $U' = \emptyset$, we define $\lambda^{X}\big(U, (\sigma_I, \sigma)\big) = \bot$.

\begin{theorem}\label{thm:synMain}
If $\varphi$ is realizable, the transducer $T$ obtained by composing $T_X$ and $T_Y$ as in step 6 of Fig.~\ref{fig:pipeline} solves the synthesis problem for $\varphi$.
\end{theorem}
An interesting corollary of the above result is that for realizable specifications with all output variables dependent, we can solve the synthesis problem in time $O(2^{k})$ instead of $\Omega(2^{k \log k})$, where $k = |A_\varphi|$.  This is because the subset construction on $A_\varphi$ suffices to obtain $T_X$, while $A_\varphi$ must be converted to a deterministic parity automaton to solve the synthesis problem in general.

\label{sec:implement}
\section{Symbolic Implementation}\label{sec:ImpNew}

In this section, we describe symbolic implementations of each of the non-shaded blocks in the synthesis flow depicted in Fig.~\ref{fig:pipeline}. Before we delve into the details, a note on the representation of NBAs is relevant.  We use the same representation as used in {\Spot}~\cite{spot} -- a state-of-the-art platform for representing and manipulating LTL formulas and $\omega$-automata.  Specifically, the transition structure of an NBA $A$ is represented as a directed graph, with nodes representing states of $A$, and directed edges representing state transitions. Furthermore, every edge from state $s$ to state $s'$ is labeled by a Boolean function $B_{(s,s')}$ over  $I \cup O$.  The Boolean function can itself be represented in several forms.  We assume it is represented as a Reduced Ordered Binary Decision Diagram (ROBDD)~\cite{bryant1995binary}, as is done in {\Spot}. Each such labeled edge represents a set of state transitions from $s$ to $s'$, with one transition for each satisfying assignment of $B_{(s,s')}$. 

\paragraph{\bfseries Implementing Algorithms~\ref{alg:all-compatible-states} and \ref{alg:find-deps-by-automaton}:} Since states of the NBA $A_\varphi$ are explicitly represented as nodes of a graph, it is straightforward to implement Algorithms~\ref{alg:all-compatible-states} and \ref{alg:find-deps-by-automaton}. The check in line 6 of Algorithm~\ref{alg:all-compatible-states} is implemented by checking the satisfiability of $B_{(s_i,s_i')}(I,O) \,\wedge\, B_{(s_j, s_j')}(I,O)$ using ROBDD operations.  Similarly, the check in line 2 of Algorithm~\ref{alg:find-deps-by-automaton} is implemented by checking the satisfiability of $\bigvee_{(s, s') \in out(p) \times out(q)} B_{(p,s)}(I,O) \wedge B_{(q,s')}(I',O') \wedge \bigwedge_{y\in Y} (y \leftrightarrow y') \wedge (z \leftrightarrow \neg z')$ using ROBDD operations. In the above formula, $I'$ (resp. $O'$) denotes a set of fresh, primed copies of variables in $I$ (resp. $O$).  

\paragraph{\bfseries Implementing transformation of $A_\varphi$ to $A_\varphi'$:}
To obtain $A_\varphi'$, we simply replace the ROBDD for $B_{(s,s')}$ on every edge $(s,s')$ of the NBA $A_\varphi$ by an ROBDD for $\exists X\, B_{(s,s')}$.  While the worst-case complexity of computing $\exists X\, B_{(s,s')}$ using ROBDDs is exponential in $|X|$, this doesn't lead to inefficiencies in practice because $|X|$ is typically small.  Indeed, our experiments reveal that the total size of ROBDDs in the representation of $A_\varphi'$ is invariably smaller, sometimes significantly, compared to the total size of ROBDDs in the representation of $A_\varphi$.  Indeed, this reduction can be significant in some cases, as the following proposition shows. 
\begin{proposition}\label{prop:exp-red-bdd}
There exists an NBA $A_\varphi$ with a single dependent output such that the ROBDD labeling its edge is exponentially (in number of inputs and outputs) larger than that labeling the edge of $A_\varphi'$.
\end{proposition}
 \begin{proof}
Let $I = \{i_1, \ldots i_n\}$ and $O = \{o_1, \ldots o_{n+1}\}$ be sets
of input and output variables.
Consider the LTL formula $\mathsf{G} \big(o_{n+1} \leftrightarrow f(i_1, \ldots i_n, o_1, \ldots o_n)\big)$, where $f$ is a Boolean function that gives the value of the $n^{th}$ least significant bit (or middle bit) in the product obtained by multiplying the unsigned integers represented by the $i_1, \ldots i_n$ and $o_1, \ldots o_n$.  The automaton $A_\varphi$ has a single state, which is also an accepting state, with a self loop labeled by an ROBDD representing $o_{n+1} \leftrightarrow f(i_1, \ldots i_n, o_1, \ldots o_n)$. It is known that this ROBDD has size in $\Omega(2^n)$ regardless of the variable ordering. However, once we detect that $o_{n+1}$ is dependent on $\{i_1, \ldots i_n\} \cup \{o_1, \ldots o_n\}$, we can project away $o_{n+1}$, and the ROBDD after projection is simply a single node representing $\mathsf{True}$. This is because for every combination of values of $\{i_1, \ldots i_n\} \cup \{o_1, \ldots o_n\}$, there is a value of $o_{n+1}$ that matches $f(i_1, \ldots i_n, o_1, \ldots o_n)$. \end{proof}

\paragraph{\bfseries Implementing transducer $T_X$:}
We now describe how to construct a Mealy machine corresponding to the transducer $T_X$.  As explained in the previous section, the transition structure of the Mealy machine is obtained by applying the subset construction to $A_\varphi$.  While this requires $O(2^{|A_\varphi|})$ time if states and transitions are explicitly represented, we show below that a sequential circuit implementing the Mealy machine can be constructed directly from $A_\varphi$ in time polynomial in $|X|$ and $|A_\varphi|$. This reduction in construction complexity crucially relies on the fact that all variables in $X$ are dependent on $I \cup (O \setminus X)$.

Let $S = \{s_0, \ldots s_{k-1}\}$ be the set of states of $A_{\varphi}$, and let $\inc{s_i}$ denote the set of states that have an outgoing transition to $s_i$ in $A_\varphi$.  To implement the desired Mealy machine, we construct a sequential circuit with $k$ state-holding flip-flops. Every state $U ~(\subseteq S)$ of the Mealy machine is represented by the state of these $k$ flip-flops, i.e. by a $k$-dimensional Boolean vector.  Specifically, the $i^{th}$ component is set to $1$ iff $s_i \in U$.  For example, if $S = \{s_0, s_1, s_2\}$ and $U = \{s_0, s_2\}$, then
$U$ is represented by the vector $\langle 1,0,1\rangle$. Let $n_i$ and $p_i$ denote the next-state input and present-state output of the $i^{th}$ flip-flop. The next-state function $\delta^X$ of the Mealy machine is implemented by a circuit, say $\Delta^{X}$, with inputs $\{p_0, \ldots p_{k-1}\} \,\cup\, I \,\cup\, (O \setminus X)$ and outputs $\{n_0, \ldots n_{k-1}\}$.  For $i \in \{0, \ldots k-1\}$,  output $n_{i}$ of this circuit implements the Boolean function $\bigvee_{s_j\,\in\, \inc{s_i}} \big(p_j \wedge \exists X\, B_{(s_j,s_i)} \big)$.
To see why this works, suppose $\langle p_0, \ldots p_{k-1}\rangle$ represents the current state $U \subseteq S$ of the Mealy machine.  Then the above function sets $n_i$ to true iff there is a state $s_j \in U$ (i.e. $p_j = 1$) s.t. there is a transition from $s_j$ to $s_i$ on some values of outputs $X$ and for the given values of $I \cup (O \setminus X)$ (i.e. $\exists X\, B_{(s_j,s_i)} \,=\, 1$).  This is exactly the condition for $s_i$ to be present in the state $U' \subseteq S$ reached from $U$ for the given values of $I \cup (O \setminus X)$ in the Mealy machine obtained by subset construction.

It is known from the knowledge compilation literature (see e.g.~\cite{darwiche-jacm,akshayCAV18,fmcad19}) that every ROBDD can be compiled in linear time to a Boolean circuit in Decomposable Negation Normal Form (\DNNF), and that every {\DNNF} circuit admits linear time projection of variables, yielding a resultant {\DNNF} circuit.  Hence, a Boolean circuit for $\exists X\, B_{(s_j,s_i)}$ can be constructed in time linear in the size of the ROBDD representation of $B_{(s_j,s_i)}$.  This allows us to construct the circuit $\Delta^X$, implementing the next-state transition logic of our Mealy machine, in time (and space) linear in $|X|$ and $|A_\varphi|$. 

Next, we turn to constructing a circuit $\Lambda^X$ that implements the output function $\lambda^X$ of our Mealy machine.  It is clear that $\Lambda^X$ must have inputs $\{p_0, \ldots p_{k-1}\} \cup I \cup (O\setminus X)$ and outputs $X$. Since $X$ is automata dependent on $I \cup (O \setminus X)$ in $A_\varphi$, the following proposition is easily seen to hold.
\begin{proposition}\label{prop:uniq-output}
Let $B_{(s,s')}$ be a Boolean function with support $I \cup O$ that labels a transition
$(s,s')$ in $A_\varphi$. For every $(\sigma_I, \sigma) \in \Sigma_I \times \Sigma_{O \setminus X}$, if $(\sigma_I, \sigma) \models \exists X\, B_{(s,s')}$, then there is a unique $\sigma' \in \Sigma_X$ such that $(\sigma_I,\sigma,\sigma') \models B_{(s,s')}$. 
\end{proposition}
Considering only the transition $(s, s')$ referred to in Proposition~\ref{prop:uniq-output}, we first discuss how to synthesize a vector of Boolean functions, say $F^{(s,s')} = \langle F_1^{(s,s')}, \ldots F_{|X|}^{(s,s')}\rangle$, where each component function has support $I \cup (O\setminus X)$, such that
$F^{(s,s')}[I \mapsto \sigma_I][O\setminus X \mapsto \sigma] = \sigma'$.  Generalizing beyond the specific assignment of $I \cup O$, our task effectively reduces to synthesizing an $|X|$-dimensional vector of Boolean functions $F^{(s,s')}$ s.t. $\forall I \cup (O\setminus X)\, \big(\exists X B_{(s,s')} ~\rightarrow~ B_{(s,s')}[X \mapsto F^{(s,s')}]\big)$ holds. 
Interestingly, this is an instance of \emph{Boolean functional synthesis} -- a problem that has been extensively studied in the recent past (see e.g.~\cite{ChakrabortyFTV22,akshayCAV18,fmcad19,ACHLTI22,AmramBFTVW21}). In fact, we know from ~\cite{fmcad19,ShahBAC21} that if $B_{(s,s')}$ is represented as an ROBDD, then a Boolean circuit for $F_{(s,s')}$ can be constructed in $\mathcal{O}\big(|X|^2.|B_{(s,s')}|\big)$ time, where $|B_{(s,s')}|$ denotes the size of the ROBDD for $B_{(s,s')}$.  For every $x_i \in X$, we use this technique to construct a Boolean circuit for $F_i^{(s,s')}$ for every edge $(s,s')$ in $A$.  The overall circuit $\Lambda^X$ is constructed such that the output for $x_i \in X$ implements the function $\bigvee_{trans. ~(s,s') ~in~ A} \big(p_s \wedge (B_{(s,s')}[X \mapsto F^{(s,s')}]) \wedge F_i^{(s,s')}\big)$.

\begin{lemma}\label{lem:lambda-imp-correct}
Let $U \subseteq S$ be a non-empty set of pairwise compatible states of $A$. For $(\sigma_I, \sigma) \in \Sigma_I \times \Sigma_{O\setminus X}$, if $\delta^X\big(U, (\sigma_I, \sigma)\big) \neq \emptyset$, then the outputs $X$ of $\Lambda^X$ evaluate to $\lambda^X\big(U, (\sigma_I,\sigma)\big)$.
In all other cases, every output of $\Lambda^X$ evaluates to $0$.
\end{lemma}
Note that $\delta^X\big(U, (\sigma_I, \sigma)\big) = \emptyset$ iff all outputs $n_i$ of the circuit $\Delta^X$ evaluate to $0$. This case can be easily detected by checking if $\bigvee_{i=0}^{k-1} n_i$ evaluates to $0$. We therefore have the following result.

\begin{theorem}
The sequential circuit obtained with $\Delta^X$ as next-state function and $\Lambda^X$ as output function is a correct implementation of transducer $T_X$, assuming (a) the initial state is $p_0 = 1$ and $p_j = 0$ for all $j \in \{1, \ldots k-1\}$, and (b) the output is interpreted as $\bot$ whenever $\bigvee_{i=0}^{k-1} n_i$ evaluates to $0$. 
\end{theorem}

\section{Experiments and Evaluation}
\label{sec:eval}
We implemented the synthesis pipeline depicted in Figure~\ref{fig:pipeline} in a tool called {\DepSynt}, using symbolic approach of Section~\ref{sec:ImpNew}. For steps 1.,  4., of the pipeline, i.e., construction of $A_\varphi$ and synthesis of $T_Y$, we used the tool {\Spot}~\cite{spot}, a widely used library for representing and manipulating NBAs. We then experimented
 \footnote{The experiment results can be accessed at \url{https://eliyaoo32.github.io/DepSynt}} 
with all available reactive synthesis benchmarks from the SYNTCOMP~\cite{syntComp} competition, a total of 1,141 LTL specifications over 31 benchmark families.

 All our experiments were run on a computer cluster, with each problem instance run on an Intel Xeon Gold 6130 CPU clocking at 2.1 GHz with 2GB memory and running Rocky Linux 8.6.
Our investigation was focussed on answering two main research questions:\\
{\bfseries RQ1:} How prevalent are dependent outputs in reactive synthesis benchmarks?\\
{\bfseries RQ2:} Under what conditions, if any, is reactive synthesis benefited by our approach, i.e., of identifying and separately processing dependent output variables?

\subsection{Dependency Prevalence}
To answer {\bfseries RQ1}, we implemented the algorithm in Section~\ref{sec:dependencies} and executed it with a timeout of 1 hour. Within this time, we were able to find 300 benchmarks out of 1,141 SYNTCOMP benchmarks, that had at least 1 dependent output variable (as per Definition~\ref{def:autdep}). Out of the 1,141 benchmarks, 260 had either timeout (41 total) or out-of-memory (219 total), out of which 227 failed because of the NBA construction (adapted from \Spot), i.e, Step 1 in our pipeline, did not terminate. We found that all the benchmarks with at least 1 dependent variable in fact belong to one of 5 benchmark families, as seen in Table~\ref{tab:deps_summary}. In order to measure the prevalence of dependency we evaluated 
(1) the number of dependent variables and
(2) the $\text{dependency ratio} = \frac{\text{Total dependent vars}}{\text{Total output vars}}$.
\begin{table}[h!]
    \centering
    \begin{tabular}{|l|r|r|r|r|}
        \hline
        \textbf{Benchmark Family} & \textbf{Total} & \textbf{Completed} & \textbf{Found Dep} & \textbf{Avg Dep Ratio}\\
        \hline
ltl2dpa & 24 & 24 & 24 & $.434$ \\
mux & 12 & 12 & 4 & $1$\\
shift & 11 & 4 & 4 & $1$\\
tsl-paper & 118 & 117 & 115 & $.46$\\
tsl-smart-home-jarvis & 189 & 167 & 153 & $.33$\\
        \hline
    \end{tabular}
    \caption{Summary for 5 benchmark families, indicating the no. of benchmarks, where the dependency-finding process was completed, the total count of benchmarks with dependent variables, and the average dependency ratio among those with dependencies.}
    \label{tab:deps_summary}

\end{table}
Out of those depicted, Mux (for multiplexer) and shift (for shift-operator operator) were two benchmark families where dependency ratio was 1.
In total, among all those where our dependency checking algorithm terminated, we found 26 benchmarks with all the output variables dependent. Of these 4 benchmarks were from Shift, 4 benchmarks from mux, 14 benchmarks from tsl-paper, and 4 from tsl-smart-home-jarvis. Note that in Mux, total dependency was found only for 4 out of 11, where the input was a power of 2.  Looking beyond total dependency, among the 300 benchmarks with at least 1 dependent variable, we found a diverse distribution of dependent variables and ratio as shown in Figure~\ref{fig:cum-deps-by-total-deps}.
\begin{figure}[h!]
    \centering
    \includegraphics[width=0.45\textwidth]{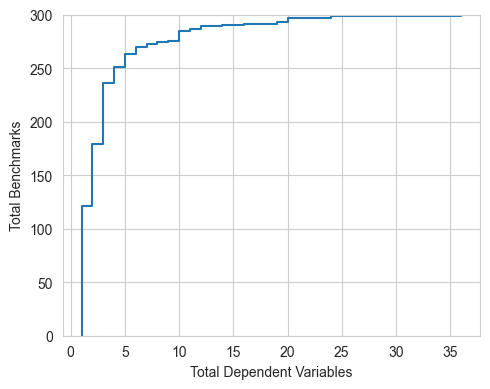}
    \includegraphics[width=0.45\textwidth]{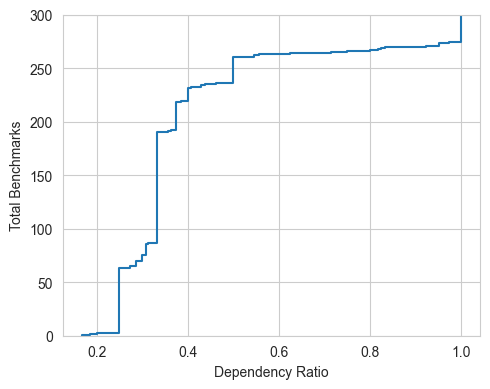}
    \caption{ 
      Cumulative count of benchmarks for each unique value of Total Dependent Variables (left plot) and Dependency Ratio (right plot). The value of $F(x)$ on y-axis represents how many benchmarks have at most $x$ (on x-axis) dependent variables/dependency ratio.}
    \label{fig:cum-deps-by-total-deps}
\end{figure}

\subsection{Utilizing Dependency for Reactive Synthesis: Comparison with other tools}
 Despite a large 1 hr time out, we noticed that most dependent variables were found within 10-12 seconds. Hence, in our tool {\DepSynt}, we limited the time for dependency-check to an empirically determined 12 seconds, and declared unchecked variables after this time as non-dependent. Since synthesis of non-dependents $T_Y$ (Step 5. of the pipeline) is implemented directly using {\Spot} APIs, the difference between our approach and {\Spot} is minimal when there are a large number of non-dependent variables. This motivated us to divide our experimental comparison, among the 300 benchmarks where at least one dependent variables was found, into benchmarks with at most 3 non-dependent variables (162 benchmarks) and more than 3 non-dependent variables (138 benchmarks). We compared {\DepSynt} with two state-of-the-art synthesis tools, that won in different tracks of SYNTCOMP23' \cite{syntComp}: 
(i) Ltlsynt (based on Spot)\cite{spot} with different configurations ACD, SD, DS, LAR, and
(ii) Strix\cite{meyer2018strix} with the configuration of BFS for exploration and FPI as parity game solver (the overall winning configuration/tool in SYNTCOMP'23). All the tools had a total timeout of 3 hours per benchmark. As can be seen from Figure~\ref{fig:depsynt_vs_tools_indepsle3}, indeed for the case of $\leq 3$ non-dependent variables, {\DepSynt} outperforms the highly optimized competition-winning tools. Even for $>3$ case, as shown in Figure~\ref{fig:depsynt_vs_tools_indepsg3}, the performance of $\DepSynt$ is comparable to other tools, only beaten eventually by Strix. 
DepSynt uniquely solved 2 specifications for which both Strix and Ltlsynt timed out after 3600s, the benchmarks are mux32, and mux64, and solved in 2ms, and 4ms respectively.
\begin{figure}[h!]
    \centering
    \begin{minipage}{0.75\textwidth}
        \centering
        \includegraphics[width=\linewidth]{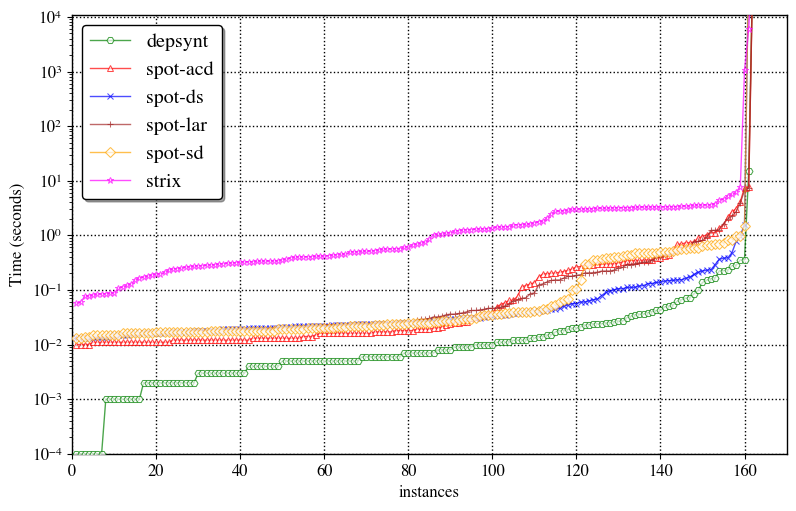}
        \caption{Cactus plot comparing DepSynt, LtlSynt, and Strix on 162 benchmarks with at most 3 non-dependent variables.}
        \label{fig:depsynt_vs_tools_indepsle3}
    \end{minipage}
    \hfill
    \begin{minipage}{0.75\textwidth}
        \centering
        \includegraphics[width=\textwidth]{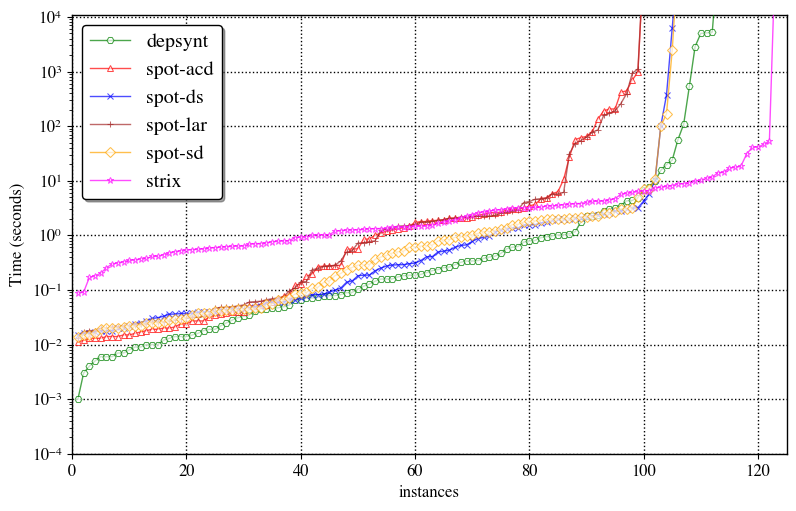}
        \caption{Cactus plot comparing DepSynt, LtlSynt, and Strix on 138 benchmarks with more than 3 non-dependent variables.}
        \label{fig:depsynt_vs_tools_indepsg3}
    \end{minipage}
\end{figure}

\subsection{Analyzing time taken by different parts of the pipeline}
\begin{figure}[h]
\centering
  \includegraphics[width=\textwidth]{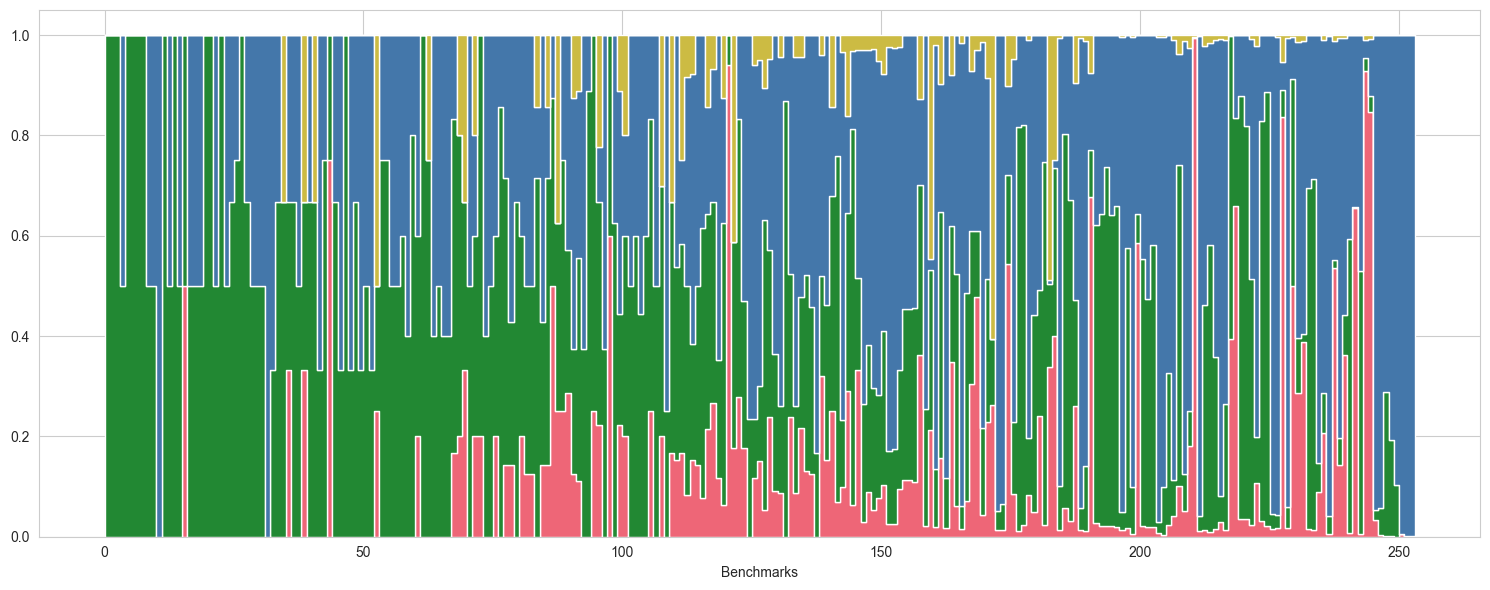}
  \caption{Normalized time distribution of DepSynt sorted by total duration over benchmarks that could be solved successfully by DepSynt.
  Each color represents a different phase of DepSynt.
  The pink is searching for dependency, the green is the NBA build, the blue is the non-dependent variables and the yellow is the dependent variables synthesis.}
  \label{fig:depsynt_time_distribution}
\end{figure}

\begin{figure}[h]
    \centering
      \includegraphics[width=0.9\textwidth]{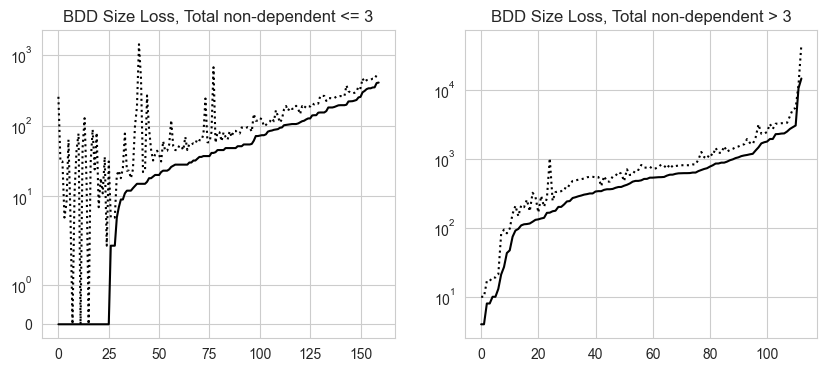}
      \caption{This figure illustrates the total BDD sizes of the NBA edges before and after the projection of the dependent variables from the NBA edges, the left figure is over the benchmark with at most 3 non-dependent variables and the right figure is over the benchmarks with 4 or more non-dependent variables.
      The solid line presents the projected BDD size and the dotted line presents the original BDD size.
      The y-axis is presented in a symmetric log-scale. The benchmarks are sorted by the projected NBA's BDD total size.}
      \label{fig:bdd-analysis-indeps-leq-3}
\end{figure}

In order to better understand where {\DepSynt} spends its time, we plotted in Figure \ref{fig:depsynt_time_distribution} the normalized time distribution of {\DepSynt}. We can see that synthesizing a strategy for dependent variables is very fast (the gray portion)- justifying its theoretical linear complexity bound, and so is the red-region depicting searching for dependency (again, a poly-time algorithm), especially compared to the white region synthesizing a strategy for the non-dependent variables. This also explains why having high dependency ratio alone does not help our approach, since even with high ratio, number of non-dependent variables could be large, resulting in worse performance overall.

\subsection{ Analysis of the Projection step (Step 3.) of Pipeline}
The rationale for projecting variables from the NBA is to reduce the number of output non-dependent variables in the synthesis of the NBA, which is the most expensive phase as Figure \ref{fig:depsynt_time_distribution} shows. To see if this indeed contributes to our better performance, we asked if projecting the dependent variables reduces the BDDs' sizes, in terms of total nodes, (the BDD represents the transitions). Figure \ref{fig:bdd-analysis-indeps-leq-3} shows that the BDDs' sizes are reduced significantly where the total of non-dependent variables is at most 3, in cases of total dependency, the BDD just vanishes and is replaced by the constant true/false. For the case of total non-dependent is 4 or more, the BDD size is reduced as well.

\subsection{An ablation experiment with Spot} As a final check, that dependency was causing the improvements seen, we conducted a control/ablation experiment where in {\DepSynt} we gave zero-timeout to find dependency, and classifies all output variables are classified as non-dependent, and called this SpotModular. As can be seen in Figure~\ref{fig:depsynt_vs_spotmodular_indepsle3}, in the left plot, for the case of benchmarks with at least 1 dependent and at most 3 non-dependent variables, the benefit of dependency-checking, while the plot on the right, shows that for the remaining cases we do not see this.
\begin{figure}[!h]
  \centering
    \includegraphics[width=0.45\linewidth]{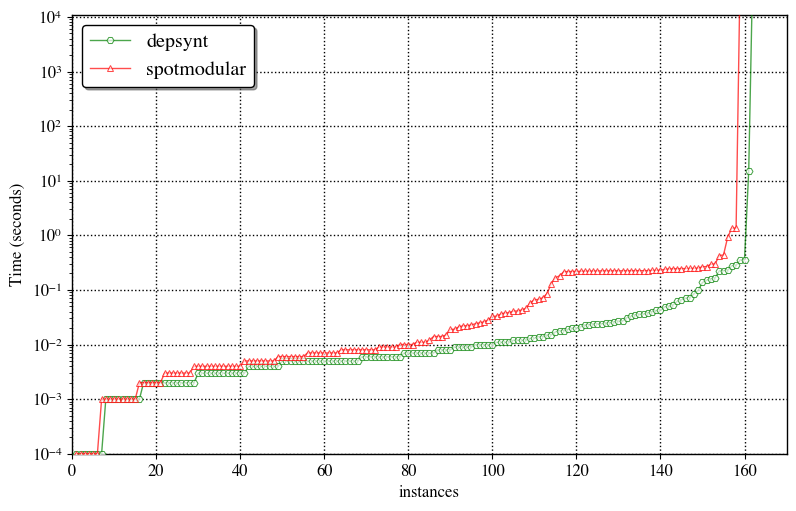}
    \includegraphics[width=0.45\textwidth]{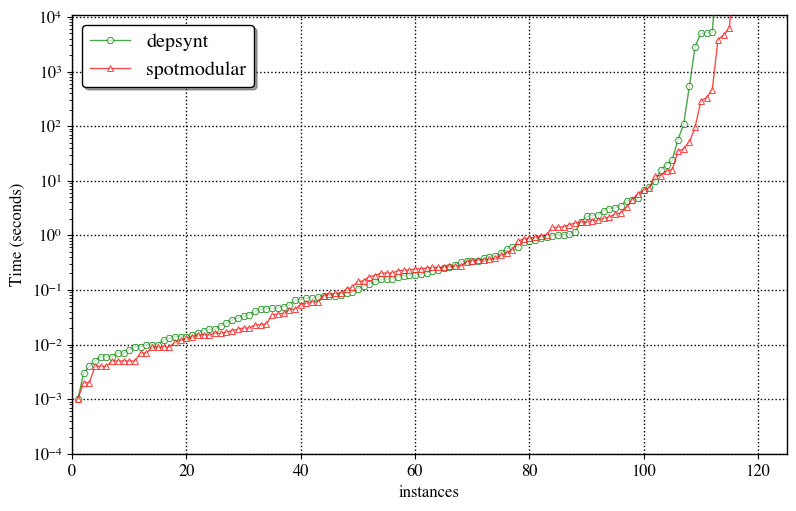}
    \caption{Cactus plot comparing DepSynt and SpotModular on 162 benchmarks with at most 3 non-dependent variables (left) and 138 benchmarks with more than 3 non-dependent variables (right).}
\label{fig:depsynt_vs_spotmodular_indepsle3}
\end{figure}

\subsection{Summary} Overall, we answered both the research questions we started with. Indeed there are several benchmarks with dependent variables, and using our pipeline does give performance benefits when no. of non-dependent variables is low. Our recipe would be to first run our poly-time check to see if there are dependents and if there aren't too many non-dependents, use our approach and otherwise switch to any existing method. 

To summarize our comparison against other tools, we state the following. When comparing with Strix, we found 252 benchmarks that had dependent variables in which DepSynt took less time than Strix. Out of which, in 126 benchmarks DepSynt took at least 1 second less than Strix. Among these, there are 10 benchmarks (shift16, LightsTotal\_d65ed84e, LightsTotal\_9cbf2546, LightsTotal\_06e9cad4, Lights2\_f3987563, Lights2\_0f5381e9, FelixSpecFixed3.core\_b209ff21, Lights2\_b02056d6, Lights2\_06e9cad4, LightsTotal\_2c5b09da) for which the time taken by DepSynt was at least 10 seconds less than that taken by Strix. These are the examples that are easier to solve by DepSynt than by Strix. For shift16, the difference was more than 1056 seconds in favor of DepSynt. Interestingly, shift16 also has all output variables dependent.

When comparing with Ltlsynt, we found 193 benchmarks that had dependent variables in which DepSynt took less time than Ltlsynt. Among these, in 27 benchmarks DepSynt took at least 1 second less than Ltlsynt. Of these, there is one benchmark (ModifiedLedMatrix5X) for which the time taken by DepSynt was at least 10 seconds less than that taken by Ltlsynt. Specifically, DepSynt took 5 seconds and Ltlsynt took 55 seconds.

\section{Conclusion}\label{sec:conc}
\vspace*{-2mm}
In this work, we have introduced the notion of dependent variables in the context of reactive synthesis. We showed that dependent variables are prevalent in reactive synthesis benchmarks and suggested a synthesis approach that may utilize these dependency for better synthesis. As part of future work, we wish to explore heuristics for choosing "good" maximal subsets of dependent variables. We also wish to explore integration of our method in other reactive synthesis tools such as Strix. 

\bibliographystyle{splncs04}
\bibliography{refs}

\clearpage
\iffullversion
\appendix

\section{Additional material for defining dependency}
\label{sec:appdep}
To find a maximal dependent set of variables we describe Algorithm~\ref{alg:findDep} called \textbf{FindDependent} as follows. \textbf{FindDependent} gets as input an LTL formula $\varphi$ over the set of variables $V=I\cup O$, and initializes the set $X$ that includes the dependent variables found so far to be empty (line $1$).
Then at every step (line $2$) we choose a variable $z\in O$ that was not tested before and check if $z$ is dependent of the remaining variables, excluding those that were already found dependent (line 3). This is done through a procedure called isDependent() which we discuss next. If so then $x$ is added  $X$ (line 4). At the end of the process, we have that $X$ is returned as an output of \textbf{FindDependent}.
The heart of the framework is naturally in the procedure \textbf{isDependent} that takes a variable $z$ and a set $Y$ and returns true iff $\{z\}$ is dependent on $Y$ in $\varphi$ according to definition~\ref{def:LTLDep}.

To prove that \textbf{FindDependent} returns a maximal dependent set, we assume that  $isDependent$ is well defined. That is, for a variable $z$ and a set $Y$ $isDependent$ returns $true$ if and only if $\{z\}$ is dependent on $Y$ in $\varphi$. The automata variant of FindDependent, that we use in this paper, makes use of the automata variant of $isDependent$, called $isAutomataDependent$, in which this assumption is verified.

\begin{theorem}\label{thm:findDep}
    Given an LTL formula $\varphi$ over the set of variables $V$, \textbf{FindDependent} returns a maximal   dependent set in $\varphi$.
\end{theorem}

\begin{algorithm}
\caption{Find Dependent Variables}\label{alg:findDep}
\hspace*{\algorithmicindent} \textbf{Input} 
LTL Formula $\varphi$ with variables $V=I\cup O$
\\
\hspace*{\algorithmicindent} \textbf{Output}
A set $X\subseteq O$ which is maximal dependent in $\varphi$

\begin{algorithmic}[1]\label{alg:isDependent}

\State $X \gets \emptyset$

\For{$z \in O$}
    \If{ isDependent($z, V\backslash (X\cup \{z\})$)}
        \State $X \gets X \cup \{ z \}$
    \EndIf
\EndFor
\State \textbf{return} $X$

\end{algorithmic}
\end{algorithm}

To prove Theorem~\ref{thm:findDep}, we first make the following claim.

\begin{claim}\label{clm:defdep2}
    Let $X'$ be a set of output variables and $X\subseteq X'$. If $X'$ is dependent in $\varphi$, then so is $X$.
\end{claim}

\begin{proof}
    Set $w,w'\in L(\varphi)$ with the same prefix $w[0,i]=w'[0,i]$ for some $i\geq 0$. Assume $w_i. (V\backslash X) = w'_i.(V\backslash X)$. Then  since $V\backslash X'\subseteq V\backslash X$ we have that $w_i.V\backslash X' = w'_i.V\backslash X'$. Therefore $w_i.X'=w_i.X'$, which means that $w_i.X=w_i.X$ as well.
\end{proof}

We now provide the proof for Theorem~\ref{thm:findDep}.

\begin{proof}[Proof of Theorem~\ref{thm:findDep}]
    Denote by $X^i$ the set $X$ in the $i$'th step of the loop obtained in line 5. We show by induction that $X^i$ is dependent on $V\backslash X^i$ in $\varphi$. For step $0$, we have that the empty set is naturally dependent on $V$. Assume by induction  that $X^{i-1}$ is dependent on $V\backslash X^{i-1}$. Then at step $i$, if no element was added to $X^{i-1}$ then $X^i=X^{i-1}$ and we are done. Otherwise, we have $X^i=X^{i-1}\cup\{z\}$. 
    Then let $j\geq 0$ be such that for two words $w,w'$ in $L(\varphi)$, $w[0,j-1]=w'[0,j-1]$ and $w_j.V\backslash X^i= w'_j.V\backslash X^i$. Then since $z$ was just added, it follows that $w_j.\{z\}=w'_j.\{z\}$. But then since $X^i=X^{i-1}\cup\{z\}$ is follows that $w_j.V\backslash X^{i-1}= w'_j.V\backslash X^{i-1}$ and by induction $w'_j.X^{i-1}=w'_j.X^{i-1}$. hence in total $w_j.X^i=w'_j.X^i$. Note that it also follows that at the end of the algorithm $X$ is dependent on $V\backslash X$ in $\varphi$ and hence $X$ is dependent in $\varphi$.
    
    To see that $X$ is maximal, assume for contradiction that it is not and let $X'\subseteq O$ be a dependent set that strictly contains $X$. Let $z\in X'\backslash X$. Then $\{z\}\subseteq X'$ hence, from the Claim above we have that $z$ is dependent in $V\backslash\{z\}$. So $z$ should have been picked up as a member of $X^i$ and therefore in $X$ in some round $i$ of the algorithm, a contradiction.\qed
\end{proof}

\fi

\end{document}